\documentclass{amsart}
\usepackage{tikz-cd}
\usepackage{latexsym,amssymb}
\usepackage[english]{babel}
\usepackage{bm}
\usepackage{amsfonts, amsthm}
\usepackage{amsmath}
\usepackage{mathrsfs}
\usepackage{multirow}
\usepackage{bm}
\usepackage{pgf,tikz}
\usepackage{graphicx}
\usepackage[caption=false]{subfig}

\usepackage{tikz-cd}
\usetikzlibrary{graphs}
\usetikzlibrary{arrows.meta}
\usetikzlibrary{shapes.geometric}

\usepackage{latexsym,amssymb,amsfonts, amsthm, amsmath}
\usepackage[english]{babel}
\usepackage{bm}
\usepackage{mathrsfs}
\usepackage{algpseudocode}
\usepackage{algorithm}

\algnewcommand{\IIf}[1]{\State\algorithmicif\ #1\ \algorithmicthen}
\algnewcommand{\EndIIf}{\unskip\ \algorithmicend\ \algorithmicif}

\newtheorem{theorem}{Theorem}
\newtheorem{lemma}{Lemma}
\newtheorem{corollary}{Corollary}

\theoremstyle{definition}

\newtheorem{remark}{Remark}

\def\F{\mathbb{F}}

\title[Upper bounds on the rate of linear $q$-ary $k$-hash codes]{Upper bounds on the rate \\ of linear $q$-ary $k$-hash codes}

\date{}
\author[S.  Della Fiore]{Stefano Della Fiore}
\author[M.  Dalai]{Marco Dalai}
\address{DI, Universit\`a degli Studi di Salerno, Via Giovanni Paolo II 132, 84084 Fisciano, Italy}
\email{sdellafiore@unisa.it}
\address{DII, Universit\`a degli Studi di Brescia,  Via Branze~43, I~25123 Brescia, Italy}
\email{marco.dalai@unibs.it}
\begin{document}
\begin{abstract}
This paper presents new upper bounds on the rate of linear $k$-hash codes in $\mathbb{F}_q^n$, $q\geq k$, that is, codes with the property that any $k$ distinct codewords are all simultaneously distinct in at least one coordinate.
\end{abstract}

\maketitle

\section{Introduction}

A $q$-ary code $C$ of length $n$ is a subset of $\{0,1,\ldots, q-1\}^n$.  We denote the rate of a $q$-ary code $C$ of length $n$ as 
$$R = \frac{1}{n} \log_q |C|\,.$$
Let $q \geq k \geq 3$ and $n \geq 1$ be integers, and let $C$ be a $q$-ary code of length $n$ with the property that for any $k$ distinct elements (codewords) we can find a coordinate in which they all differ. A subset $C$ with this property is called $(q,k)$-hash code of length $n$. In particular, $(q,3)$-hash codes are known as $q$-ary trifferent codes (or just trifferent codes when $q=3$).  The problem of finding upper and lower bounds for the maximum size of $(q,k)$-hash codes is a fundamental problem in theoretical computer science and information theory. It appears, as the name suggests, in the study of families of perfect hash functions and in the study of the zero-error capacity of some discrete channels with list decoding, see \cite{fredman-komlos, korner-marton, arikan94, blackburn1998optimal, della2022improved} for more details (see also \cite{bhandari-2022} for a related problem).

An elementary double counting argument, as shown in \cite{korner-marton}, gives the following bound on the cardinality of $(q,k)$-hash codes:
\begin{equation}\label{eq:recursivebound}
    |C| \leq  (k-1) \left(\frac{q}{k-1}\right)^n \text{ for every } q \geq k \geq 3\,.
\end{equation}
In 1984 Fredman and Koml\'os \cite{fredman-komlos} improved the bound in \eqref{eq:recursivebound} for every $q = k \geq 4$ and sufficiently large $n$, obtaining the following result:
\begin{equation}\label{eq:fredmanKomlos}
    |C| \leq \left((q-k+2)^{q^{\underline{k-1}} / {q^{k-1}}}\right)^{n+o(n)}\,,
\end{equation}
where $q^{\underline{k-1}} = q (q-1) \cdots (q-k+2)$.
Fredman and Koml\'os also provided, using standard probabilistic methods, the following lower bound:
\begin{equation}\label{eq:lbfk}
    |C| \geq \left(\left(1 - \frac{q^{\underline{k}}}{q^k}\right)^{-1/(k-1)}\right)^{n+o(n)}\,.
\end{equation}

A generalization of the upper bounds given in equations \eqref{eq:recursivebound} and \eqref{eq:fredmanKomlos} was derived by K\"{o}rner and Marton \cite{korner-marton} in the form
\begin{equation}\label{eq:kornerMarton}
    |C| \leq \min_{0\leq j\leq k-2} \left(\left(\frac{q-j}{k-j-1}\right)^{q^{\underline{j+1}} / {q^{j+1}}}\right)^{n+o(n)}\,.
\end{equation}

In 1998, Blackburn and Wild \cite{blackburn1998optimal} (see also \cite{bassalygo1997}) improved the bound of K\"orner and Marton for every $q$ sufficiently larger than $k$, proving
\begin{equation}\label{eq:blackWild}
    |C| \leq (k-1) q^{\lceil \frac{n}{k-1} \rceil}\,.
\end{equation}

Much effort has been spent during the years to refine the bounds given in \eqref{eq:recursivebound}, \eqref{eq:fredmanKomlos} and \eqref{eq:kornerMarton}. See for example \cite{della2022maximum, kurz2024trifferent} and the recent breakthrough \cite {bhandari2024improved} for case of $q=k=3$, \cite{arikan1994, arikan94, dalai2019improved} in case of $q=k=4$,  \cite{costa2021new} in case of $q=k=5,6$, and \cite{korner-marton, guruswami2022beating, della2021new, della2022improved} for $q \geq k \geq 5$. However, to the best of our knowledge, for $q$ sufficiently larger than $k$ no improvements over the upper bound \eqref{eq:blackWild} have been obtained. 

For the sake of completeness, it is worth noting that some improvements on the lower bound given in \eqref{eq:lbfk} on the largest size of $(q,k)$-hash codes have been recently obtained in~\cite{xing2023beating} for $q \in [4,15]$,  all integers between $17$ and $25$, and for a sufficiently large $q$. For $q=k=3$ the best known lower bound is due to  K\"orner and Marton~\cite{korner-marton} and it is equal to $(9/5)^{n/4 +o(1)}$.

In contrast, no exponential improvement has been made on the simple bound given in \eqref{eq:recursivebound} for $q= k = 3$ (see for example \cite{costa2021gap} for a discussion on the intractability of this problem even with some recent powerful techniques such as the slice-rank method). Until recently, only improvements on the multiplicative constant had been obtained (see \cite{della2022maximum, kurz2024trifferent}), while a polynomial improvement has been obtained in a beautiful recent work by Bhandhari and Kheta \cite {bhandari2024improved}.

However, if we restrict the codes to be linear, i.e., we require $C$ to be a linear subspace of $\mathbb{F}_3^n$, exponential improvements have been recently obtained. Upper bounds on the rate of linear trifferent codes have been considered first in \cite{Pohoata-Zakharov-2022}, where it was proved that for some $\epsilon>0$,
\begin{align}
|C| & \leq 3^{\left(\frac{1}{4}-\epsilon\right)n} \label{eq:lin_triff_first} \\ &\approx 1.3161^n\,. \nonumber
\end{align}
This result was then improved in \cite{bishnoi-etal-2023}, where connections with minimal codes were used to show that
\begin{align}
|C| & \leq 3^{n/4.5516+o(n)} \label{eq:lin_triff_best} \\ &\approx 1.2731^{n}, \nonumber
\end{align}
where the constant $1.2731$ is the numerical solution of an equation that we will explain in Section~\ref{sec:k=q=3}. The authors also showed that there exist linear trifferent codes of length $n$ and size  $\frac{1}{3}(9/5)^{n/4}$,  matching,  asymptotically in $n$,  the best known lower bound on trifferent codes (without the linearity constraint) obtained in \cite{korner-marton}.

We note that when $q$ is small compared to $k>3$, no linear $k$-hash codes of dimension $2$ exist. In particular, Blackburn and Wild \cite{blackburn1998optimal} showed that this is true for every $q \leq 2k - 4$ (so, in particular, the case $q=k>3$ is of no interest). The authors in \cite{ng2001k} improved this result for $k \geq 9$ showing that when $q$ is a square, $\sqrt{q} > 5$, no linear $k$-hash codes of dimension $2$ exist whenever $q \leq \left(\frac{k-1}{2}\right)^2$. Hence in these regimes, we know that linear $k$-hash codes in $\mathbb{F}_q^n$ are relatively \textit{simple objects} since their asymptotic rates are equal to zero.

In this paper we provide upper bounds on the rate of linear $k$-hash codes in $\mathbb{F}_q^n$ for general values of $q\geq k\geq 3$. These are the first known (non-trivial) such bounds. For $q=k=3$ they recover the best known result of \cite{bishnoi-etal-2023} given in equation~\eqref{eq:lin_triff_best}. Also, in the range of $q$ much larger than $k$, they improve the general bound of equation \eqref{eq:blackWild} (in terms of code rate as $n\to \infty$).

\section{A simpler proof for $q=k=3$}
\label{sec:k=q=3}
In this section, we present a re-derivation of \eqref{eq:lin_triff_first} and \eqref{eq:lin_triff_best} by a straightforward application of a method already presented\footnote{The main tool used is essentially some form of Jamison's bound both here and in \cite{Pohoata-Zakharov-2022, bishnoi-etal-2023}. The difference is mainly a matter of how this tool can be combined with other ideas in coding theory.} in \cite{calderbank-etal-1993}. This simpler approach is the starting point for the extension to the general case $q\geq k \geq 3$ which is then presented in the next section.

The idea is to modify slightly the proof of \cite[Corollary 2.1]{calderbank-etal-1993}. The main tool to be used is Jamison's bound \cite{jamison-1977}.
\begin{lemma}[\cite{jamison-1977}]
Let $q\geq 3$ be a prime power, and let $\mathcal{H}$ be a set of hyperplanes in $\F_q^m$ whose union is $\F_q^m\setminus\{0\}$. Then $|\mathcal{H}|\geq (q-1)m$.
\end{lemma}

Let $C$ be a linear trifferent code of dimension $m$ and length $n$. Let $G$ be the $m\times n$ generator matrix, let $d$ be the minimum Hamming distance of the code, and let $x$ be a codeword of weight $d$. Finally call $u\in \mathbb{F}_3^m$ the information vector associated to $x$, that is, assume $x=uG$, and also assume without loss of generality (by appropriate sorting and re-scaling of the columns of $G$) that $x$ has $0$s in the last $n-d$ coordinates and $1$s in the first $d$ coordinates. Then, since the code is trifferent, any codeword different from $0$ and $x$ must have a coordinate equal to $2$ among the first $d$ ones. So, if we call $g_i$ the $i$-th column of $G$, the $d$ affine subspaces defined by
\begin{equation*}
H_i=\{v\in \F_3^m : v\cdot g_{i}=2\}\,,i=1,\ldots, d
\end{equation*}
cover the set $\F_3^m$ with the exception of $0$ and $u$. Adding another subspace $H_{d+1}=\{v\in \F_3^m : v\cdot g_{1}=1\}$ we also covers $u$, still leaving out $0$. By Jamison's bound, $d+1\geq 2 m$. In terms of rates and relative minimum distance $\delta=d/n$ this becomes
\begin{equation}
R\leq \frac{1}{2}\delta + o(1)\,.
\label{eq:bound_R_delta}
\end{equation}
The rest comes from known upper bounds on the minimum Hamming distance of codes. Using the Plotkin bound 
$$
\delta \leq \frac{2}{3}(1-R)+o(1)
$$
gives $R\leq (1-R)/3 + o(1)$, which is asymptotically $R\leq 1/4$, essentially equivalent to\footnote{Strictly speaking, to obtain the positive $\epsilon$ in \eqref{eq:lin_triff_first} we need the fact that the Plotkin bound is not tight at positive rates.} \eqref{eq:lin_triff_first}. The stronger bound \eqref{eq:lin_triff_best} is obtained instead by using the best known bound on $\delta$, which is the linear programming bound of \cite{mceliece-et-al-1977} adapted to $q$-ary codes \cite{aaltonen1990new}, defined implicitly in $\delta$ by the inequality
\begin{equation}
R \leq H_q\left(\frac{1}{q}\left(q - 1 - (q-2)\delta - 2\sqrt{(q-1)\delta(1 - \delta)}\right)\right)
\label{eq:MRRW}
\end{equation}
with $q=3$ and 
\begin{equation}\label{eq:Hq}
H_q(t) = t \log_q(q-1) - t\log_q t - (1-t)\log_q(1-t)\,.
\end{equation}
The bound on $R$ is found by combining \eqref{eq:bound_R_delta} and \eqref{eq:MRRW}, which means solving  \eqref{eq:MRRW} for equality with $\delta=2R$.

\section{General Case $q\geq k\geq 3$}
Our extension to general $q\geq k\geq 3$ is essentially based on the idea of iterating the technique used for $q=k=3$. To do this, we will need to consider a generalized notion of distance among tuples of codewords and a generalization of Jamison's bound to multiple coverings.
The latter is given by this result of Bruen \cite{bruen-1992}.
\begin{lemma}[\cite{bruen-1992}]
\label{lemma:bruen}
Let $\mathcal{H}$ be a multiset of hyperplanes in $\F_q^m$. If no hyperplane in $\mathcal{H}$ contains $0$ and each point in $\F_q^m\setminus \{0\}$ is covered by at least $t$ hyperplanes in $\mathcal{H}$, then
$$
|\mathcal{H}|\geq (m+t-1)(q-1)\,.
$$
\end{lemma} 

We need to introduce the following technical lemma.

\begin{lemma}
Let $C$ be a linear code of dimension $m$ in $\F_q^d$, let $x_1, x_2, \ldots, x_\ell \in C$ be $\ell \leq q-1 \leq m$ linearly independent codewords which are all pairwise distinct in each coordinate and contain no zeros, and let $C'$ be a subcode of $C$ of dimension $m - \ell$ that intersects trivially, only in the origin, the subspace spanned by the $x_i$'s. For $i=1,\ldots, d$, set $S_i=\F_q\setminus\{0,x_{1,i},x_{2,i},\ldots,x_{\ell,i}\}$.

Assume that, for each $c\in C' \setminus \{0\}$, we have $c_i\in S_i$ for at least $t$ values of $i$. Then 
\begin{equation}
m - \ell \leq \frac{q-\ell-1}{q-1}d-t+1\,.
\label{eq:m_vs_d&t}
\end{equation}
\label{lem:multicover}
\end{lemma}
\begin{proof}
Let $G$ be the (full rank) $(m - \ell) \times d$ generator matrix of the subcode $C'$ and $g_i$ its $i$-th column. Consider the hyperplanes 
$$
H_{i,b}=\{v\in \F_q^{m-\ell} \mid v\cdot g_i = b\}\,,\quad i=1,\ldots,d\,,\  b\in S_i\,.
$$
Since $|S_i|=q-(\ell+1)$, these are $(q-\ell-1)d$ hyperplanes none of which contains the zero vector. On the other hand, by assumption, each $v \in \mathbb{F}_q^{m-\ell} \setminus \{0\}$ is covered at least $t$ times by those hyperplanes. Therefore from Lemma~\ref{lemma:bruen} we then have
$$
(q-\ell-1)d\geq (m - \ell +t-1)(q-1)
$$
which is equivalent to the statement.
\end{proof}

We are now ready to state our main result. 

\begin{theorem}\label{thm:main}
Let $C$ be a linear $k$-hash code in $\F_q^n$ of rate $R=m/n$ and relative distance $\delta$. Then,
\begin{equation}
R\leq \frac{\delta}{\sum_{i=1}^{k-2}\frac{(q-1)^i}{(q-2)^{\underline{i}}}} + o(1)
\label{eq:main_th}
\end{equation}
where $(q-2)^{\underline{i}}=(q-2)(q-3)\cdots(q-i-1)$.
\label{th:main_th}
\end{theorem}

\begin{proof}
We will show that if the rate exceeds the claimed bound we can find, by means of an iterative process, a collection of $k$ codewords $\{0,x_1\ldots, x_{k-1}\}$ which do not satisfy the $k$-hash property. Figure \ref{fig:codewords} gives a graphical representation of the properties we will require for the codewords. We start with a codeword $x_1$ of minimum weight $d=\delta_1 n= \delta n$, where we assume without loss of generality that $x_1$ is non-zero in the first $d$ coordinates. Any set of $k$ codewords in $C$ which includes $0$ ad $x_1$ cannot satisfy the $k$-hash property in the last $n-d$ coordinates, so we can focus on the first $d$ coordinates and consider the punctured code, call it $C_{[d]}$. Note that puncturing is injective. Indeed, if two distinct codewords $y,y'\in C$ are equal in $[d]$, then the codewords $0,x_1,y-y'$ are all distinct and are not $3$-hashed. Hence, the code is not a $k$-hash code for any $k\geq 3$. This means that i) $C_{[d]}$ is also an $m$-dimensional subspace in $\F_q^d$ and ii) we can refer to \emph{codewords} without ambiguity as to whether we mean in $C$ or $C_{[d]}$.

We now want to select a codeword $x_2$ which is linearly independent of $x_1$ and matches either with $0$ or with $x_1$ in many coordinates. Consider thus the linear subspace of $C_{[d]}$ of dimension $m$. We now use Lemma~\ref{lem:multicover}. Take $t$ which contradicts \eqref{eq:m_vs_d&t} with $\ell=1$, that is such that
\begin{align*}
\frac{m-1}{n}> \frac{q-2}{q-1}\cdot\frac{d}{n}-\frac{t}{n}+\frac{1}{n}\,.
\end{align*}
Setting $\delta_2=t/n$, this means taking $\delta_2\in [0,1]$ such that
\begin{align*}
R > \frac{q-2}{q-1}\delta_1-\delta_2+o(1)\,.
\end{align*}

\begin{figure}[t]
\centering
\includegraphics{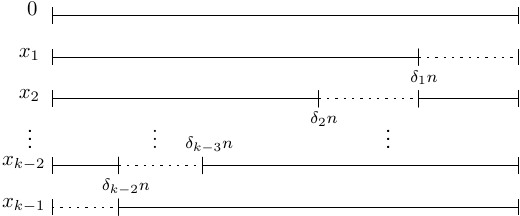}
\caption{Codewords used in the proof of Theorem \ref{th:main_th}. Each codeword collides with one of the previous codewords in each coordinate of the dotted part.}
\label{fig:codewords}
\end{figure}

Lemma~\ref{lem:multicover} then implies that there is a codeword $x_2$ linearly independent of $x_1$ such that $x_{2,i}\notin \{0,x_{1,i}\}$ for less than $\delta_2 n$ coordinates in $[1,\ldots,\delta_1 n]$. Assume without loss of generality that these are the first coordinates. Then, the three codewords $\{0,x_1,x_2\}$ are not $3$-hashed in any of the last $(1-\delta_2)n$ coordinates. This means that any $k$ codewords which include $0,x_1,x_2$ must be $k$-hashed in one of the first $\delta_2 n$ coordinates. Again we can restrict our attention on the punctured code $C_{[\delta_2 n]}$. Note that again the puncturing is injective, for the same reason already mentioned in the first iteration; if distinct codewords $y$ and $y'$ coincide over $[\delta_2 n]$, then $0,\ x_1,\ x_2,\ y-y'$ are $4$ distinct codewords which are not $4$-hashed (note that $x_1$ and $x_2$ cannot equal $y-y'$ since they are non-zero in $[\delta_2 n]$ while $y-y'$ is zero there). Thus $C_{[\delta_2 n]}$ has dimension $m$. Furthermore, $x_1$ and $x_2$ remain linearly independent also when restricted to the coordinates $[\delta_2 n]$. Indeed, if we assume by contradiction that in $[\delta_2 n]$ we have $x_2=\alpha x_1$ for some $\alpha \in \mathbb{F}_q \setminus \{0\}$, then the codeword in $C_{[\delta_1 n]}$ defined by $y=\alpha x_1$ coincides with $x_2$ in $[\delta_2 n]$ and so, by injectivity of our code restriction to $[\delta_2 n]$, necessarily $y=x_2$, which is impossible since $x_2$ and $x_1$ were chosen to be linearly independent in $[\delta_1 n]$. Therefore, since the subspace $C_{[\delta_2 n]}$ has dimension $m$, we can iterate our procedure invoking Lemma~\ref{lem:multicover} with $\ell=2$.

Continuing this way, at iteration $j$ we have $j$ linearly independent codewords $x_1,\ldots,x_j$ in $C_{[\delta_j n]}$ such that 
$$|\{0,x_{1,i},\ldots,x_{j,i}\}|\leq j \text { for all } i>\delta_{j}n\,,$$ and we find a $(j+1)$-th linearly independent codeword over the coordinates $[\delta_j n]$ which is $(j+1)$-hashed with the previous $j$ codewords and 0 only in the first $\delta_{j+1}n$ coordinates, where $\delta_{j+1}$ is chosen to satisfy
\begin{equation*}
R > \frac{q-j-1}{q-1}\delta_j-\delta_{j+1}+o(1) \,,
\end{equation*}
that is
\begin{equation}
\delta_{j+1}>\frac{q-j-1}{q-1}\delta_j-R.
\label{eq:j-th_iter}
\end{equation}
The restriction of the code to the first $\delta_{j+1}n$ coordinates is again injective and thus also preserves the linear independence of $x_1,\ldots,x_{j+1}$, because if one of those codewords was a linear combination of the other ones over $[\delta_{j+1}n]$ then it would coincide with the same linear combination taken over $[\delta_j n]$ which contradicts the linear independence of $x_1,\ldots,x_{j+1}$ over $[\delta_j n]$.

We iterate this for $j=1,\ldots,k-3$, finding $x_1,\ldots,x_{k-2}$ such that  
$$|\{0,x_{1,i},\ldots,x_{k-2,i}\}|\leq k-2 \text{ for all } i>\delta_{k-2}n$$ 
while 
$$|\{0,x_{1,i},\ldots,x_{k-2,i}\}|=k-1 \text{ for all } i\leq \delta_{k-2}n.$$
At this point we can find one last, linearly independent codeword $x_{k-1}$, such that 
$$
x_{k-1,i}\in \{0,x_{1,i},\ldots,x_{k-2,i}\} \,,\ \text{for all } i\leq \delta_{k-2}n
$$
if \eqref{eq:j-th_iter} is satisfied for $j=k-2$ with $\delta_{k-1}=0$, that is if
\begin{equation}
R > \frac{q-k+1}{q-1}\delta_{k-2}\,.
\label{eq:stop_iter}
\end{equation}
This gives us $k-1$ codewords $x_1,\ldots,x_{k-1}$ such that the $k$ codewords $\{0,x_1,\ldots,$ $x_{k-1}\}$ are not $k$-hashed in any coordinates.
The condition on $R$ can be obtained by using recursively equation~\eqref{eq:j-th_iter} in \eqref{eq:stop_iter} with initialization $\delta_1=\delta$. This leads to
\begin{equation*}
\frac{q-1}{q-k+1}R > \frac{(q-2)^{\underline{k-3}}}{(q-1)^{k-3}}\delta - R \sum_{j=0}^{k-4}\frac{(q-k+j+1)^{\underline{j}}}{(q-1)^j}
\end{equation*}
which, after rearrangements of the terms, is equivalent to $R$ violating \eqref{eq:main_th}.
\end{proof}

\section{New upper bounds}

In this section, we provide, using the result obtained in Theorem~\ref{thm:main}, new upper bounds on the rate of linear $k$-hash codes in $\mathbb{F}_q^n$ for every $q \geq k \geq 3$.

Using the Plotkin bound for $q \geq 3$ we obtain the following corollary.
\begin{corollary}\label{cor:Plotkin}
    Let $C$ be a linear $k$-hash code in $\mathbb{F}_q^n$ of rate $R$. Then,
    \begin{equation}\label{eq:plot}
        R \leq \left(1 + \frac{q}{q-1} \sum_{i=1}^{k-2} \frac{(q-1)^{i}}{(q-2)^{\underline{i}}}\right)^{-1} + o(1)\,.
    \end{equation}
\end{corollary}
\begin{proof}
By the Plotkin bound we have that a code of length $n$ with relative minimum distance $\delta$ and rate $R$ satisfies, for $n$ large enough, the inequality $R \leq 1- \frac{q}{q-1} \delta$ which implies that
\begin{equation}\label{eq:plotkinB}
    \delta \leq \frac{q-1}{q} (1-R)\,.
\end{equation}
Now, using the upper bound on $\delta$ of equation~\eqref{eq:plotkinB} and the bound for linear $k$-hash codes given in Theorem~\ref{th:main_th} we obtain
$$
    R \leq \frac{q-1}{q} \frac{(1-R)}{\sum_{i=1}^{k-2} \frac{(q-1)^{i}}{(q-2)^{\underline{i}}}} + o(1)\,.
$$
Therefore, rearranging the terms we obtain the statement of the corollary.
\end{proof}

As done for the case $q = k = 3$ in Section~\ref{sec:k=q=3}, we can use the first linear programming bound of \cite{aaltonen1990new} to obtain the following corollary.

\begin{corollary}\label{cor:Aaltonen}
    Let $C$ be a linear $k$-hash code in $\mathbb{F}_q^n$ of rate $R$. Then,
    $$
        R \leq \frac{\delta^{*}}{\sum_{i=1}^{k-2}\frac{(q-1)^i}{(q-2)^{\underline{i}}}} + o(1)\,,
    $$
    where $\delta^{*}$ is the unique root of the following equation in $x$
    \begin{equation*}
        \frac{x}{\sum_{i=1}^{k-2}\frac{(q-1)^i}{(q-2)^{\underline{i}}}} = H_q\left(\frac{1}{q}\left(q - 1 - (q-2)x - 2\sqrt{(q-1)x(1 - x)}\right)\right)\,,
    \end{equation*}
    where $0 \leq x \leq \frac{q-1}{q}$ and $H_q$ is the function defined in equation~\eqref{eq:Hq}.
\end{corollary}

In Table~\ref{Tab:PVA}, we compare the bounds provided in Corollaries~\ref{cor:Plotkin}, \ref{cor:Aaltonen} and the one given in \eqref{eq:kornerMarton} for $q \in [3, 64]$. It can be seen that the linear programming bound performs better for $q \leq 19$ while for $q \geq 23$ the Plotkin bound gives a better result.

\begin{table}[t!]
\caption{Upper bounds on the rate of linear $3$-hash codes in $\mathbb{F}_q^n$ for a prime power $q \in [3,64]$. All numbers are rounded upwards.}
\label{Tab:PVA}
\centering
\footnotesize
\setlength{\tabcolsep}{6pt} 
\renewcommand{\arraystretch}{1} 
\begin{tabular}{ |l|l|l|l| } 
\hline
 $q$ & Corollary \ref{cor:Plotkin} & Corollary \ref{cor:Aaltonen} & Equation \eqref{eq:kornerMarton} \\
 \hline
$3$ & $1/4 = 0.25$ & 0.2198 & 0.3691\\ 
$4$ & $1/3 = 0.\overline{3}$ & 0.3000 & $1/2 = 0.5$ \\
$5$ & $3/8 = 0.375$ & 0.3441 & 0.5694 \\ 
$7$ & $5/12 = 0.41\overline{6}$ & 0.3928 & 0.6438 \\ 
$8$ & $3/7 = 0.\overline{428571}$ & 0.4080 & $2/3 = 0.\overline{6}$\\ 
$9$ & $7/16 = 0.4375$ & 0.4200 & 0.6846\\
$11$ & $9/20 = 0.45$ & 0.4373 & 0.7110 \\ 
$13$ & $11/24 = 0.458\overline{3}$ & 0.4497 & 0.7298 \\ 
$16$ & $7/15 = 0.4\overline{6}$ & 0.4628 & $3/4 = 0.75$\\ 
$17$ & $15/32 = 0.46875$ & 0.4663 & 0.7554 \\
$19$ & $17/36 = 0.47\overline{2}$ & 0.4721 & 0.7646\\ 
$23$ & $21/44 = 0.477\overline{27}$ & 0.4811 & 0.7790\\
$25$ & $23/48 = 0.4791\overline{6}$ & 0.4846 & 0.7847\\
$27$ & $25/52 = 0.48\overline{076923}$ & 0.4877 & 0.7897\\
$29$ & $27/56 = 0.482\overline{142857}$ & 0.4903 & 0.7942\\
$31$ & $29/60 = 0.48\overline{3}$ & 0.4927 & 0.7982\\
$32$ & $15/31 \approx 0.483871$ & 0.4938 & $4/5 = 0.8$\\
$37$ & $35/72 = 0.486\overline{1}$ & 0.4984 & 0.8081\\
$41$ & $39/80 = 0.4875$ & 0.5013 & 0.8134\\
$\cdots$ & $\cdots$ & $\cdots$ & $\cdots$ \\
$64$ & $31/63 = 0.\overline{492063}$ & 0.5119 & $5/6 = 0.8\overline{3}$\\
 \hline
\end{tabular}
\end{table}

\begin{remark}
    We observe that one could use the second linear programming bound or the straight-line bounds given in \cite{aaltonen1990new, laihonen1998upper} to improve the results of Corollaries~\ref{cor:Plotkin} and \ref{cor:Aaltonen} for different values of $q$ and $k$. Here we avoid to show those improvements to keep a simpler presentation of our results.
\end{remark}

For a fixed value of $k$ and $q \to \infty$, both bounds given in Corollaries~\ref{cor:Plotkin} and \ref{cor:Aaltonen} converge to $1/(k-1)$, which is the same upper bound on the rate of $(q,k)$-hash codes (not necessarily linear) that one can derive from equation~\eqref{eq:blackWild}. Corollary~\ref{cor:Plotkin} approaches $1/(k-1)$ from below since the rhs of equation~\eqref{eq:plot} is strictly increasing in $q$ and since $\sum_{i=1}^{k-2} \frac{(q-1)^i}{(q-2)^{\underline{i}}} \geq k-2$, while Corollary~\ref{cor:Aaltonen} does not, see for example Table~\ref{Tab:PVA} where for $k=3$ and $q = 41, \ldots, 64$ we have upper bounds that exceed $1/2$.

We can compare the bound of Corollary~\ref{cor:Plotkin} with the one given in equation~\eqref{eq:kornerMarton} to obtain the following theorem.

\begin{theorem}\label{thm:comp}
    For every $q \geq k^2$ and $k \geq 4$, the bound of Corollary~\ref{cor:Plotkin} improves the one of K\"orner and Marton for general codes given in equation~\eqref{eq:kornerMarton}.
\end{theorem}
\begin{proof}
We need to show that
\begin{equation}\label{eq:comp}
    \left(1 + \frac{q}{q-1} \sum_{i=1}^{k-2} \frac{(q-1)^{i}}{(q-2)^{\underline{i}}}\right)^{-1}  < \min_{0\leq j\leq k-2} \frac{q^{\underline{j+1}}}{q^{j+1}} \log_q \left( \frac{q-j}{k-j-1} \right)\,.
\end{equation}
We lower bound the rhs of equation~\eqref{eq:comp} as follows
\begin{equation*}
    \min_{0\leq j\leq k-2} \frac{q^{\underline{j+1}}}{q^{j+1}} \log_q \left( \frac{q-j}{k-j-1} \right)
    \geq \frac{q^{\underline{k-1}}}{q^{k-1}} \log_q \left( \frac{q}{k-1} \right)
    \geq \frac{1}{2} \left( \frac{q-k+2}{q} \right)^{k-2}\,,
\end{equation*}
since the function $\log_x (\alpha x)$ is increasing in $x$ for $x \geq 2$ and $0 < \alpha < 1$ and since $k \geq 4$ and $q \geq k^2$. Then, by Bernoulli's inequality we have that
$$
    \frac{1}{2} \left( \frac{q-k+2}{q} \right)^{k-2} \geq \frac{1}{2}\left(1 - \frac{(k-2)^2}{q}\right)\,.
$$
Since the lhs of \eqref{eq:comp} is less than $1/(k-1)$, in order to prove the statement of the theorem we just need to show that
$$
    \frac{1}{k-1} \leq \frac{1}{2}\left(1 - \frac{(k-2)^2}{q}\right)\,,
$$
but this inequality is satisfied for ${q \geq k^2}$ and $k \geq 4$.
\end{proof}

\begin{figure}[!t]
\centering
\includegraphics[scale = 0.98]{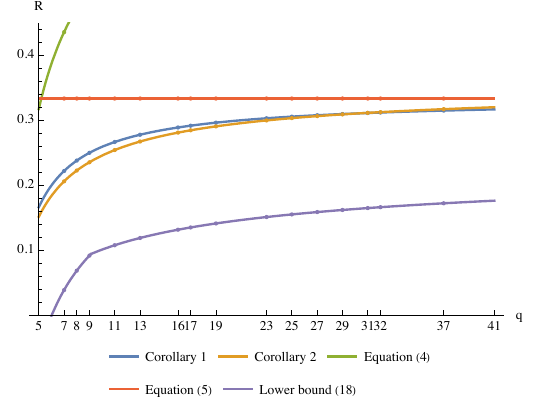}
\caption{Comparison between upper and lower bounds for $q \geq 5$ and $k=4$.}
\label{fig:comptriff}
\end{figure}

We conjecture that Theorem~\ref{thm:comp} still holds also if we relax the hypothesis to $q \geq 2k-3$ and $k\geq 3$. This would imply that, for all the interesting values of $q$ and $k$ (since the asymptotic rate of linear $k$-hash codes in $\mathbb{F}_q^n$ for $q \leq 2k-4$ is zero), our bound provides the best result.

In support of our conjecture, Table~\ref{Tab:PVA} provides an instance for $k=3$ where for every $q \geq 3$ the bound of Corollary~\ref{cor:Plotkin} improves the one of equation~\eqref{eq:kornerMarton} and Figure~\ref{fig:comptriff} reports the comparison between our bounds and the best known bounds in the literature for $k=4, 5$ and $q \geq 2k-3$. In addition, we have numerically verified the conjecture for every $k \in [3, 100]$ and $q \geq 2k-3$.

The authors in~\cite{bassalygo1997}, using classical random coding techniques, provide the following lower bound on the rate of linear $k$-hash codes in $\mathbb{F}_q^n$ for $q \geq \binom{k}{2}$:

\begin{equation}\label{eq:linlb}
    R \geq \min\bigg\{-\frac{1}{k-1} \log_q\left( 1 - \frac{q^{\underline{k}}}{q^k} \right),  \frac{1}{k-2} \left( 1-  \log_q \binom{k}{2} \right)\bigg\} + o(1)\,,
\end{equation}
where it can be seen that for $q$ sufficiently larger than $k$ the minimum of \eqref{eq:linlb} is achieved by the first term. This implies that for such values of $q$ and $k$, the lower bound coincides with the one for general codes given in equation \eqref{eq:lbfk}. However, there is still a large gap between upper and lower bounds. 

In Figure~\ref{fig:comptriff}, we compare our upper bounds and the lower bound of equation \eqref{eq:linlb} for $k=4,5$ and $q \geq 2k-3$. We note that both the upper bounds of Corollaries \ref{cor:Plotkin}, \ref{cor:Aaltonen} and the lower bound \eqref{eq:linlb} are asymptotically equal to $1/(k-1)$ for a fixed value of $k$ as $q \to \infty$.

\section*{Acknowledgements}
The authors would like to thank Lakshmi Prasad Natarajan for pointing out an error in the original derivation of the main result of this paper. Following his comments, the proof is now also simpler, and the result slightly stronger.
The authors would also like to thank Simone Costa for useful discussions on this topic.


\end{document}